\documentclass[11pt]{article}
\newcommand{\ceil}[1]{\lceil #1 \rceil}
\newcommand{\floor}[1]{\lfloor #1 \rfloor}
\usepackage{amssymb}
\usepackage{amsthm}
\usepackage{graphicx}

\newtheorem{lemma}{Lemma}
\newtheorem{conjecture}{Conjecture}

\title{3-Colorable Delaunay Triangulations}
\author{Lucas Moutinho Bueno \\\small{Institute of Mathematics and Computer Sciences}\\\small{University of S\~ao Paulo - S\~ao Carlos - Brazil}}

\begin{document}

\maketitle

\section{Introduction}

\subsection{Objective}

The input of the problem we are trying to solve is a set $X$ of $n$ two-dimensional points.
The output is a 3-colorable two-dimensional Delaunay triangulation $T$ for $X \cup Y$,
where $Y$ is a set of $m$ new points. We want to $m$ be as few as possible. 

\subsection{Motivation}

Delaunay triangulations are popular triangulations that maximize the minimum angle of all the interior angles of the triangles.
This property is desirable for some geometric operations such as interpolation or rasterization.

The adjacency of the triangles in a triangulation might be represented by a data structure.
We are particularly interested in using the GEM data structure \cite{GEM1,GEM2} for this purpose.

This structure is compact and has some operational advantages over others. It is also generalizable to dimensions higher than two. However, it can only
represent 3-colorable triangulations , i.e., the triangulations that one can assign one of colors (labels) from $\{0,1,2\}$ to each vertex as long as 
any two adjacent vertices have different colors.

The GEM data structure represents each triangle $t$ by a record with pointer fields $t.p[i]$, $i \in \{0,1,2\}$, corresponding to three colors. The field
$t.p[i]$ points to the record of the triangle $t'$ that is adjacent to $t$ across the edge $e$ opposite to the vertex of color $i$.

Unlike other triangulation data structures, there is no need to identify which side of $t'$ is the edge $e$, or to perform runtime checks to obtain that information
(e.g: from vertex pointers): in the GEM structure $t'.p[i]$ points back to $t$.

\subsection{Definitions}

A \emph{pseudo-triangulation} (on the plane) is a partition of a compact connected space $\mathcal{S} \subset \mathbb{R}^2$ into sets of vertices ($\mathcal{V}$), edges ($\mathcal{E}$) and triangles ($\mathcal{T}$).

A \emph{triangulation} (on the plane) is a pseudo-triangulation, such that:

\begin{itemize}

	\item Every edge is an open line segment of $S$;
	\item Every vertex is a point of $S$;
	\item The endpoints of every edge are vertices of the triangulation;
	\item Every vertex is an endpoint of some edges;
	\item Every triangle is bounded by a cycle of three edges and vertices;
	\item Every edge is on the boundary of one or two triangles;
	\item $\bigcup \mathcal{V}\mathcal{E}\mathcal{T} = \mathcal{S}$;
	\item $\bigcap \mathcal{V}\mathcal{E}\mathcal{T} = \emptyset$.

\end{itemize}

An edge that is in the boundary of two triangles is an \emph{interior edge}, otherwise it is a \emph{border edge}.

If all edges with a common endpoint $v$ are interior edges and $v$ is the endpoint of at least one edge, then we say that $v$ is an \emph{interior vertex}, otherwise $v$ is a \emph{border vertex}.

A  \emph{Delaunay triangulation} (DT) is a triangulation where $\mathcal{S}$ is the convex hull of $\mathcal{V}$ and where no vertex is inside the circumcircle of any triangle. 
If there is no subset of $\mathcal{V}$  with more than three co-circular vertices, then the DT is unique.

An \emph{even Delaunay triangulation} (EDT) is a triangulation where every interior vertex is even (every interior vertex is the endpoint of an even number of edges).

A \emph{locally Delaunay} edge is an edge that is either a border edge or there is no vertex inside the circumcircles of both its incident triangles.

An \emph{incomplete Delaunay triangulation} (IDT) is a pseudo-triangulation where every edge is locally Delaunay.

If $C$ is a cycle with four edges and four vertices  and $e$ is the edge that splits $C$ into two triangles, then the \emph{flip edge} of $e$ is an edge that would have the endpoints of $C$ different from $e$.

Given a vertex $v$ of a triangulation, $v(x)$ is the horizontal coordinate of $v$ in the Euclidian plane and $v(y)$ is its vertical coordinate.

\section{Algorithm for 3-colorable Delaunay triangulations}

It is known that a two-dimensional Delaunay triangulation is 3-colorable if, and only if, all its interior vertices are even~\cite{3-color}. Therefore constructing a 3-colorable two-dimensional Delaunay triangulation is
the same as constructing an EDT.

\subsection{Review of the Divide and Conquer algorithm for DTs}

Below, we describe the Divide and Conquer Algorithm (DQA) for two-dimensional DT from Guibas and Stolfi~\cite{divide}. This algorithm is
the base to construct an EDT and it runs in $O(n\ log(n))$ time.

The input of the algorithm is a set $X$ of $n$ two-dimensional points. Its output is a Delaunay triangulation $T$, where $ \mathcal{V} = X$.

\paragraph{$DQA (X)$:}

\begin{enumerate}

	\item $S_X \leftarrow Sort(X)$

	\item $\{T_1, \ldots T_{\ceil{\frac{n}{3}}}\} \leftarrow Split(S_X)$

	\item $T \leftarrow Merge(\{T_1, \ldots T_{\ceil{\frac{n}{3}}}\}, \ceil{\frac{n}{3}})$. Return $T$

\end{enumerate}

The procedure $Sort(X)$ sorts the points of $X$ into a sequence $S_X$ of vertices in ascending order of the horizontal coordinate.

The procedure $Split(S_X)$ splits the sequence of vertices $S_X$ into $\floor{\frac{n}{3}}$ groups of three consecutive vertices each and one group of two vertices if $mod(n,3) = 2$ or a single vertex if $mod(n,3) = 1$.
For each group of three vertices, a triangle is created while connecting the vertices with edges. If there is a group of two vertices, an edge is created between them. Each group $i$ of vertices, together with its new triangle and edges,
is stored as the pseudo-triangulation $T_i$, for every $i = 1  \ldots \ceil{\frac{n}{3}}$. Note that $T_i$ is always an IDT. In particular, $T_i$ is a DT for $i = 1  \ldots \floor{\frac{n}{3}}$.

The procedure $Merge(T_1, \ldots, T_k, k)$ merges $k$ IDTs ($T_1$ to $T_k$) into a single DT.
The steps of the procedure are described below.
It uses the following auxiliaries: three IDTs $T_L$ (left triangulation), $T_R$ (right triangulation) and $T'$,
five vertices $v_L$, $v_R$, $v_{LH}$, $v_{RH}$ and $w$, and an edge $e_{LR}$.

\paragraph{$Merge(T_1, \ldots, T_k, k)$:}

\begin{enumerate}

	\item if $k = 1$ then return $T_1$ \label{st.test1}

	\item if $k > 2$ then:  \label{st.test2}

	\begin{enumerate}

		\item $T_L \leftarrow Merge(T_1, \ldots, T_{\floor{\frac{k}{2}}}, \floor{\frac{k}{2}})$

		\item $T_R \leftarrow Merge(T_{\floor{\frac{k}{2}}+1}, \ldots, T_k, \ceil{\frac{k}{2}})$

	\end{enumerate}

	\item Set $T' \leftarrow \emptyset$. Let $v_L$ and $v_R$ be the vertices of $T_L$ and $T_R$ with the lowest vertical coordinate, respectively. 
	Let $v_{LH}$ and $v_{RH}$ be the vertices of $T_L$ and $T_R$ with the highest vertical coordinate, respectively. \label{st.firstset}

	\item Create the edge $e_{LR}$ with endpoints $v_L$ and $v_R$ and add $e_{LR}$ into $T'$ \label{st.firstedge}

	\item while $v_L \neq v_{LH}$ and $v_R \neq v_{RH}$, repeat:  \label{st.loop}

	\begin{enumerate}

		\item find the vertex $w$ from $T_L \cup T_R$ such that there is no vertex from $T_L$ or $T_R$ inside the circle formed by $w$, $v_L$ and $v_R$ \label{st.findw}

		\item if $w \in T_L$ then $v_L \leftarrow w$ else $v_R \leftarrow w$

		\item create the edge $e_{LR}$ with endpoints $v_L$ and $v_R$, that will form a new triangle $t$; add $e_{LR}$ and $t$ into $T'$

		\item remove from $T_L$ and $T_R$ all edges intersecting $e_{LR}$, together with their incident triangles \label{st.endloop}

	\end{enumerate}

	\item Return the triangulation formed by $T_L \cup T_R \cup T'$ \label{st.final}
	
\end{enumerate}

The proof of correctness of the DQA, in particular the existence and uniqueness of the vertex $w$ of Step~\ref{st.findw}, can be found in~\cite{divide}.
It is also important to note that during the execution of the Merge procedure,  $T_L \cup T_R \cup T'$ is an IDT.

See Figure~\ref{fig:delaunay_ex} for an example of execution of the Merge procedure.

\begin{figure}[h]
\centering
\includegraphics[width=360pt]{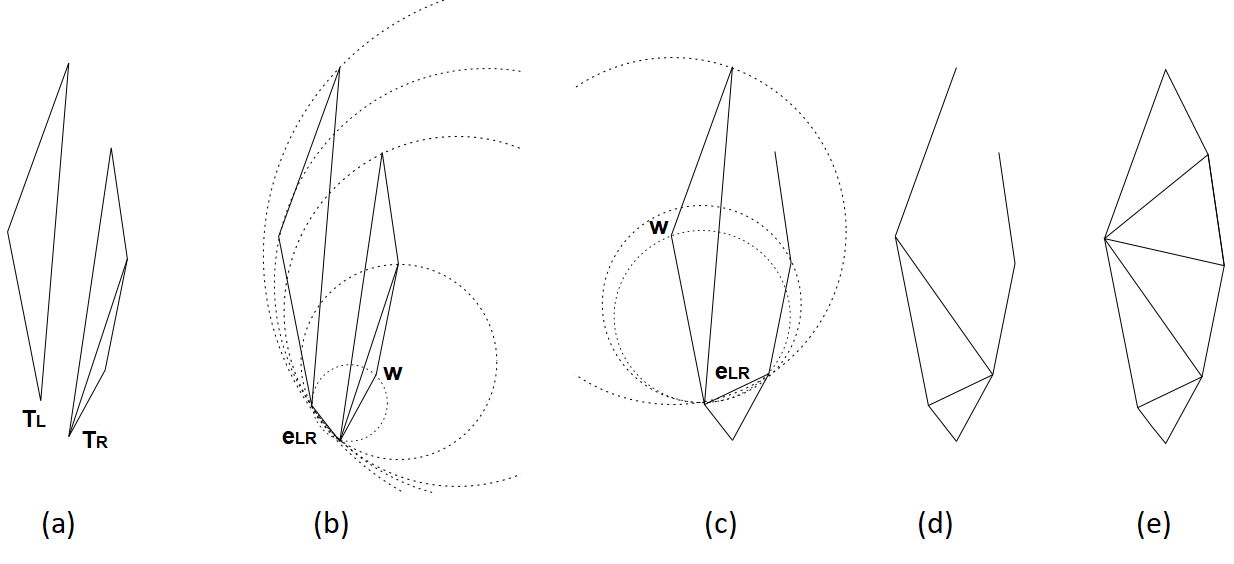}
\caption{Example of execution of the Merge procedure from the DQA. 
It starts with the triangulations $T_L$ and $T_R$ (a); it adds the edge $e_{LR}$ and it searches for the vertex $w$ to create the next edge (b)(c); it removes some edges along the way (c)(d); it returns a DT (e).}
\label{fig:delaunay_ex}
\end{figure}

\subsection{Making the Triangulation Even}

In order to make an EDT, we modify the Merge procedure with the following steps, that are executed just after Step~\ref{st.findw}.

\begin{enumerate}

		\item[\ref{st.findw}']  if $w \in T_L$ and $v_L$ is odd or if $w \in T_R$ and $v_R$ is odd, then:
				 
		\begin{enumerate}
		
			\item[\ref{st.findw}'1] create three edges $e_w$, $e_L$ and $e_R$ and a vertex $u$ as a common endpoint of these edges.
				 Set $w$, $v_L$ and $v_R$ as the second endpoint of $e_w$, $e_L$ and $e_R$, 
				 respectively. Add $u$, $e_w$, $e_L$ and $e_R$ into $T'$. Two new triangles will be created and also added into $T'$. 
				 The coordinates of $u$ must be chosen in such a way that $T_L \cup T_R \cup T'$ is an IDT.

			\item[\ref{st.findw}'2] Let $u_{R}$ be the vertex of $T_R$ with the lowest horizontal coordinate. 
				If $u(x) < u_{R}(x)$ then $T_L \leftarrow Merge\_u(T_L,u,e_L)$ and $v_L \leftarrow u$,
				else $T_R \leftarrow Merge\_u(u,T_R,e_R)$ and $v_R \leftarrow u$.
				
			\item[\ref{st.findw}'3] Continue from Step~\ref{st.loop}.
		
		\end{enumerate}

\end{enumerate}

We will call the new modified procedure as Merge' and the new modified algorithm as DQA'. 

The procedure $Merge\_u$ is similar to the procedure $Merge'$, except for the following items:

\begin{itemize}

\item The parameter $k$ of $Merge'$ is replaced by an edge $e$ in $Merge\_u$.

\item Steps~\ref{st.test1} and \ref{st.test2} of $Merge'$ are skipped in $Merge\_u$.

\item  In Step~\ref{st.firstset} of $Merge\_u$, $v_L$ and $v_R$ are set to be the endpoints of the edge $e$ (one of them will be the vertex $u$, passed as parameter).

\item The edge $e$  won't be created again in Step~\ref{st.firstedge} but will be included into $T'$ normally.

\end{itemize}

In practice, one can see that when $Merge\_u(T_L,u,e_L)$ is executed, $u = T_R$. Therefore, $u = v_R = v_{HR}$ during the execution of the loop (Step~\ref{st.loop}) and the vertex $w$ will always be chosen from the vertices of $T_L$ at Step~\ref{st.findw}.
Analogously, one can see a similar execution pattern for $Merge\_u(u,T_R,e_R)$, by changing $T_R$ by $T_L$ and other related variables.

See Figure~\ref{fig:delaunay_e_ex} for an example of execution of DQA'. 

\begin{figure}[hp]
\centering
\includegraphics[width=300pt]{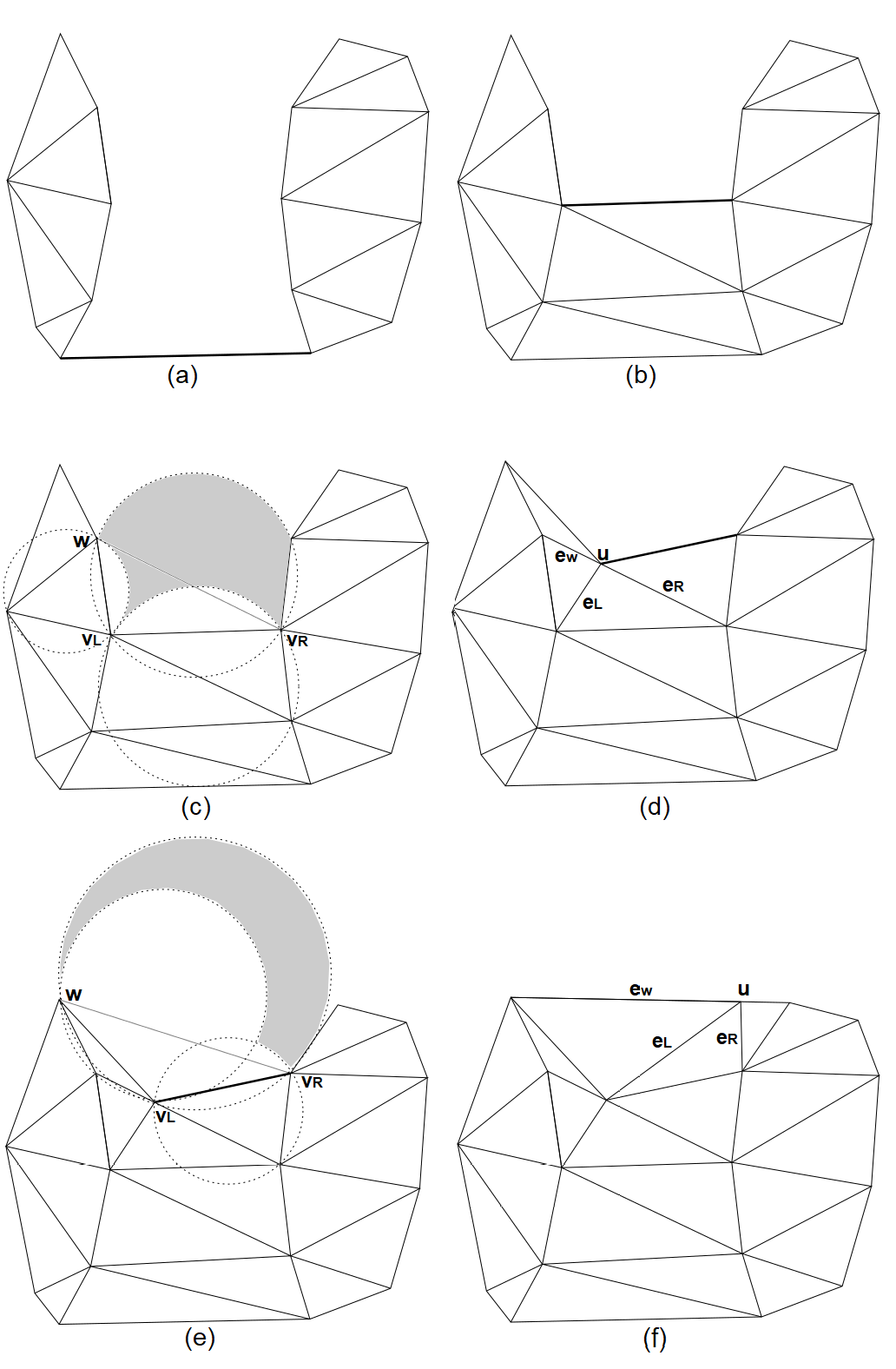}
\caption{Example of execution of the Merge' procedure from the DQA'. The edge $e_{LR}$ is drawn thicker.
Merge' is executed as Merge (a)(b) until $v_L$ is odd (c)(e). The vertex $u$ can be created anywhere in the shaded area. 
The Steps~\ref{st.findw}'1 and~\ref{st.findw}'2 are executed as well as one additional iteration of  Step~\ref{st.loop} (d)(f).}
\label{fig:delaunay_e_ex}
\end{figure}

\section{Correctness}

First of all, note that the vertex  $u_{R}$ of Step~\ref{st.findw}'2 of $Merge'$ and $Merge\_u$ serves to ensure that, after execution of $Merge\_u$, all vertices of $T_L$ will have a lower or the same horizontal coordinate than any vertex from $T_R$,
preserving a property needed to $Merge$  function properly, as well as $Merge'$ and $Merge\_u$.

During the execution of DQA' and, in particular, during the execution of $Merge'$ and $Merge\_u$, one can see that every edge added into $T'$ can only be added if $T_L \cup T_R \cup T'$ is an IDT.
This is true even for the edge added during Step~\ref{st.firstedge}, since it is a border edge and its addition into $T'$ won't create any interior edges in  $T_L \cup T_R \cup T'$. 
Therefore, assuming that the algorithm is correct, $T_L \cup T_R \cup T'$ must always be an IDT.

One can also see that all vertices from the input of DQA' start as border vertices and all vertices created at Step~\ref{st.findw}'1 of $Merge'$ and $Merge\_u$ are created as border vertices. 
The Step~\ref{st.findw} of $Merge'$ and $Merge\_u$ is the only step that turns a border vertex into an interior vertex and Step~\ref{st.findw}'1  prevents that vertex to be odd.
Therefore, assuming that the algorithm is correct, the interior vertices of $T_L \cup T_R \cup T'$ must always be odd.

Finally, note that the conditions for ending the loop at Step~\ref{st.loop} are the same for $Merge$, $Merge'$ and $Merge\_u$. Those conditions turn $T_L \cup T_R \cup T'$ into a triangulation of a convex hull and therefore turn an IDT into a DT.
Since the interior vertices of $T_L \cup T_R \cup T'$ must always be odd, assuming that the algorithm is correct,  the final triangulation must be an EDT.

However, we still need to guarantee the existence of the vertex $u$ created at Step~\ref{st.findw}'1 of $Merge'$ and $Merge\_u$  and we still need to guarantee that the creation of new vertices won't cause the loop at Step~\ref{st.loop} to be endless.
The lemma below proves half of what we still need. 

\begin{lemma}

During the execution of Step~\ref{st.findw}'1 of $Merge'$ or $Merge\_u$, it is always possible to choose coordinates for $u$ in such a way that $T_L \cup T_R \cup T'$ is an IDT.

\end{lemma}

\begin{proof}

To prove this lemma, we need to guarantee that there are coordinates for $u$ where the new interior edge $e_w$ will be locally Delaunay and the former border edges  $e_L$ and $e_R$ will remain locally Delaunay after the execution of  Step~\ref{st.findw}'1.

The new edges $e_L$ and $e_R$ will be locally Delaunay because they will be on the border of  $T_L \cup T_R \cup T'$.
The remaing edges that are not $e_w$ don't need to be tested because they were already locally Delaunay and won't change with Step~\ref{st.findw}'1.

 Let $e_L'$ be the edge of $T_L$ with endpoints $w$ and $v_L$, and let $v_L'$ be the vertex of $T_L$ opposite from $e_L'$. Analogously,
 let $e_R'$ be the edge of $T_R$ with endpoints $w$ and $v_R$, and let $v_R'$ be the vertex of $T_R$ opposite from $e_R'$. 
 We also define $C_L$ as the cycle formed by the vertices $v_L$, $v_w$ and $v_L'$; $C_R$ as the cycle formed by the vertices $v_R$, $v_w$ and $v_R'$; 
 and $C_{LR}$ as the cycle formed by $v_L$, $v_R$ and $e_{LR}$.

 We must chose coordinates for $u$ in such a way that they won't be inside $C_L$ nor $C_R$.
 Since we want the edge $e_w$ to be in the final triangulation instead of the edge $e_{LR}$, then $u$ must be inside $C_{LR}$ or part of it.

We know that $T_L \cup T_R \cup T' \cup e_{LR}$ is an IDT.  Therefore, $v_L$ is either outside or a point of $C_R$. Likewise, $v_R$ is either outside or a point of $C_L$. 

 Let $A_{LR}$ be the arc of $C_{LR}$ from $v_L$ to $v_R$.
 If $v_L$ is a point of $C_R$, then $v_L$, $v_R$, $v_w$ and $v_R'$ are cocircular and $C_R = C_{LR}$. Likewise, if $v_L$ is a point of $C_L$, then $C_L = C_{LR}$.
 On the other hand, if $v_L$ is outside $C_R$, then the points of $A_{LR}$ are outside $C_R$, except for $v_R$.
 Likewise, if $v_R$ is outside $C_L$, then the points of $A_{LR}$ are outside $C_L$, except for $v_L$.
 In any case, since the length of $A_{LR}$ is greater than zero, there is a point of $A_{LR}$, different from $v_L$ and $v_R$, that is either outside or part of $C_R$ and/or $C_L$.
 Since that point is part of $C_{LR}$ it can be chosen for $u$ in such a way that $T_L \cup T_R \cup T'$ remains an IDT after the execution of Step~\ref{st.findw}'1.

\end{proof}

%% e se A_{LR} interceptar $T_L$ ou $T_R$ ???????

Note that the vertex $u$ does not need to be a point of $A_{LR}$ (as assumed in the proof of the lemma above). It can be any point inside $C_{LR}$ that is also outside or part of $C_L$ and $C_R$.\\

Now we need to prove that $DQA'$ ends. For this purpose, we need to show that $Merge\_u$ is called a finite number of times for each new vertex $u$ created by $Merge'$.
We also need to show that, for each new vertex created by $Merge'$, the loop at Step~\ref{st.loop} won't be called again at least some of the times for any instance of the problem.

For the first part, we note that since $Merge\_u(T_L,u,e_L)$ will always have $u=T_R$, then any odd vertex tested at Step~\ref{st.findw}' will be part of $T_L$. 
Also, that vertex will never be tested again during a recursion of $Merge\_u$, shrinking the number possible vertices to be tested at each recursive call of $Merge\_u$, ending the recursion eventually.

Unfortunately, for the second part, we couldn't prove that the vertices created by $Merge'$ will always prevent the loop to be called a finite number of times. On the other hand, we didn't find any instance where the loop is endless. We then created the following conjecture.

\begin{conjecture}

 During the execution of $Merge'$, Step~\ref{st.findw}'1 is called a finite number of times.
 
\end{conjecture}

If the conjecture above is proved then we can conclude that the DQA' always works. Otherwise the DQA' works for the instances that we have seen.

\section{Final remarks}

As future work we intend to prove the conjecture of this paper by finding an upper bound for the number of new vertices created by DQA', or disprove it by finding an instance where the DQA' doesn't stop.

If we don't succeed to prove an upper bound for DQA' or if the bound achieved grows exponentially on the number of input vertices, we could develop an algorithm that creates an approximated 3-colored Delaunay triangulation.
We could start by changing the condition at Step~\ref{st.findw}' of $Merge'$ to the following:\newline

if $w \in T_L$ and $v_L$ is odd and $v_L \in X$ or if $w \in T_R$ and $v_R$ is odd and $v_R \in X$, then:\newline

\noindent where $X$ is the set of the input vertices of  DQA'. In this case the DQA' would execute Steps~\ref{st.findw}'1 and~\ref{st.findw}'2 $O(|X|)$ times, creating at most $|X|$ vertices.
Despite the output triangulation being Delaunay, some interior vertices could be odd at the end of the execution.
In order to make all interior vertices even and therefore making the triangulation 3-colorable we could, for a start,
use the output of the DQA' with the modified condition above as the input of the algorithm of Bueno and Stolfi~\cite{BnS} that subdivides any triangulation into a 3-colored triangulation.
The algorithm of Bueno and Stolfi may bisect some triangles so that the smallest angle of the final triangulation would be grater or equal to half the smallest angle of the input Delaunay triangulation.
However, we hope that we will find a better solution for this problem in the future.

I would like to thank Professor Antonio Castelo for reviewing this paper.


\begin{thebibliography}{}

\bibitem{GEM1} S. Lins, A. Mandel,  “Graph-encoded 3-manifolds”, Discrete Mathematics 57(3):261-284, 1985.

\bibitem{GEM2} L. M. Bueno, “Colored triangulations of maps”, doctoral dissertation, University of Campinas, 2016.

\bibitem{divide} L Guibas, J Stolfi, “Primitives for the manipulation of general subdivisions and the computation of Voronoi”, ACM Transactions on Graphics 4(2):74-123, 1985.

\bibitem{3-color} K. Diks, L. Kowalik, M. Kurowski, “A new 3-color criterion for planar graphs”,  28th  International  Workshop  on  Graph-Theoretic  Concepts  in  Computer Science:138–149, 2002.

\bibitem{BnS} L. Bueno, J. Stolfi, “3-Colored Triangulation of 2D Maps”, International Journal of Computational Geometry and Applications 26(02):111-133, 2016.

\end{thebibliography}
\end{document}